\documentclass[journal]{IEEEtran}

\usepackage{graphicx}
\usepackage{bm}
\usepackage{enumerate}
\usepackage{float}
\usepackage{multicol}

\usepackage{color}
\usepackage{url}

\usepackage{cite}

\usepackage{amsmath}
\usepackage{amsthm}

\usepackage{amssymb}

\newtheorem{theorem}{Theorem}
\newtheorem{corollary}{Corollary}

\usepackage{algorithmic}
 \usepackage{algorithm}

\usepackage{array}

\ifCLASSOPTIONcompsoc
  \usepackage[caption=false,font=normalsize,labelfont=sf,textfont=sf]{subfig}
\else
  \usepackage[caption=false,font=footnotesize]{subfig}
\fi

\begin{document}
\captionsetup{belowskip=0pt,aboveskip=0pt}

\title{Power Allocation in HARQ-based \\Predictor Antenna Systems}

\author{Hao~Guo,~\IEEEmembership{Student~Member,~IEEE},
        Behrooz~Makki,~\IEEEmembership{Senior~Member,~IEEE},
        \\Mohamed-Slim Alouini,~\IEEEmembership{Fellow,~IEEE},
        and Tommy~Svensson,~\IEEEmembership{Senior~Member,~IEEE}
\thanks{H. Guo and T. Svensson are with the Department of Electrical Engineering, Chalmers University of Technology, 41296 Gothenburg, Sweden (email: hao.guo@chalmers.se; tommy.svensson@chalmers.se).}
\thanks{B. Makki is with Ericsson Research, 41756 Gothenburg, Sweden (email: behrooz.makki@ericsson.com).}
\thanks{M.-S. Alouini is with the King Abdullah University of Science and Technology,
Thuwal 23955-6900, Saudi Arabia (e-mail: slim.alouini@kaust.edu.sa).}}


\maketitle


\begin{abstract}
In this work, we study the performance of predictor antenna (PA) systems using hybrid automatic repeat request (HARQ). Here, the PA system is referred to as a system with two sets of antennas on the roof of a vehicle. In this setup, the PA positioned in the front of the vehicle can be used to predict the channel state information at the transmitter (CSIT) for data transmission to the receive antenna (RA) that is aligned behind the PA. Considering spatial mismatch, due to the vehicle mobility, we derive closed-form expressions for the optimal power allocation and the minimum average power of the PA systems under different outage probability constraints. The results are presented for different types of HARQ protocols and we study the effect of different parameters on the performance of PA systems. As we show, our proposed approximation scheme enables us to analyze PA systems with high accuracy. Moreover, for different vehicle speeds, we show that HARQ-based feedback can reduce the outage-limited power consumption of PA systems by orders of magnitude.
\end{abstract}

%
\IEEEpeerreviewmaketitle

\vspace{-6mm}
\section{Introduction}
Vehicle communication is one of the important use cases in the fifth generation of wireless networks (5G) and beyond \cite{Dang2020what}.  Here, the focus is to provide efficient and reliable connections to cars and public transports, e.g., busses and trains. Channel state information at the transmitter (CSIT), plays an important role in achieving these goals, as it enables advanced closed-loop transmission schemes such as link adaptation, multi-user scheduling, interference coordination and spatial multiplexing schemes.  However,  typical CSIT acquisition systems, which are mostly designed for (semi)static channels, may not work well as the speed of the vehicle increases. This is because, depending on the vehicle speed, the position of the antennas may change quickly and the CSIT becomes inaccurate. 

To overcome this issue, \cite{Sternad2012WCNCWusing} proposes the concept of predictor antenna (PA).  Here, in its standard version, a  PA system is referred to as a setup with two (sets of) antennas on the roof of a vehicle. The PA positioned in the front of the vehicle can be used to improve the CSIT for data transmission to the receive antenna (RA) that is aligned behind the PA. The potential of such setups have been previously shown through experimental tests \cite{Sternad2012WCNCWusing,Dinh2013ICCVEadaptive,Jamaly2014EuCAP}, and its performance has been analyzed in, e.g., \cite{Guo2019WCLrate,guo2020semilinear,guo2020rate}. 

One of the challenges of the PA setup is spatial mismatch that  causes CSIT for the RA to be partially inaccurate. This occurs if the RA does not reach the same spatial point as the PA, due to, e.g.,  the delay for preparing the data  is not equal to the time that is needed until the RA reaches the same point as the PA \cite{Dinh2013ICCVEadaptive}.  On the other hand, in a typical PA setup the spectrum  is underutilized, and the spectral efficiency could be further improved in case the PA could be used not only for channel  prediction, but also for data transmission. We address these challenges by implementing hybrid automatic repeat request (HARQ)-based protocols in PA systems as follows.

In this work, we analyze the outage-limited performance of PA systems using HARQ.   With our proposed approach, the PA is used not only for improving the CSIT in the retransmissions to the RA, but also for data transmission in the initial round. In this way, as we show,  the combination of PA and HARQ protocols makes it possible to improve the spectral efficiency, and adapt the transmission parameters to mitigate the effect of spatial mismatch.

The problem is cast in the form of minimizing the average transmission power subject to an outage probability constraint. Particularly, we  develop approximation techniques to derive closed-form expressions for the instantaneous and average transmission power as well as the optimal power allocation minimizing the outage-limited power consumption. The results are presented for the cases with different repetition time diversity (RTD) and incremental redundancy (INR) HARQ protocols \cite{6566132,5754756}. Moreover, we study the effect of different parameters such as the antennas separation and the vehicle speed on the system performance. 

As we show through analysis and simulations, the implementation of HARQ as well as power allocation can improve the outage-limited performance of PA systems by orders of magnitude, compared to the cases with no retransmission. For example, consider an outage probability constraint of  $10^{-5}$, initial rate $R=4$ nats-per-channel-use (npcu) and a maximum of two transmission rounds. Then, compared to the cases with no retransmission, our proposed power-adaptive PA-HARQ scheme can reduce the required signal-to-noise ratio (SNR) by  25 dB and 30 dB for the RTD and the INR schemes, respectively.

\vspace{-4mm}
\section{Problem Formulation}

Here, we first introduce the basics of PA systems which is followed by our proposed HARQ-based PA setup.
\vspace{-3mm}
\subsection{Standard PA System}
\begin{figure}
\centering
  \includegraphics[width=0.8\columnwidth]{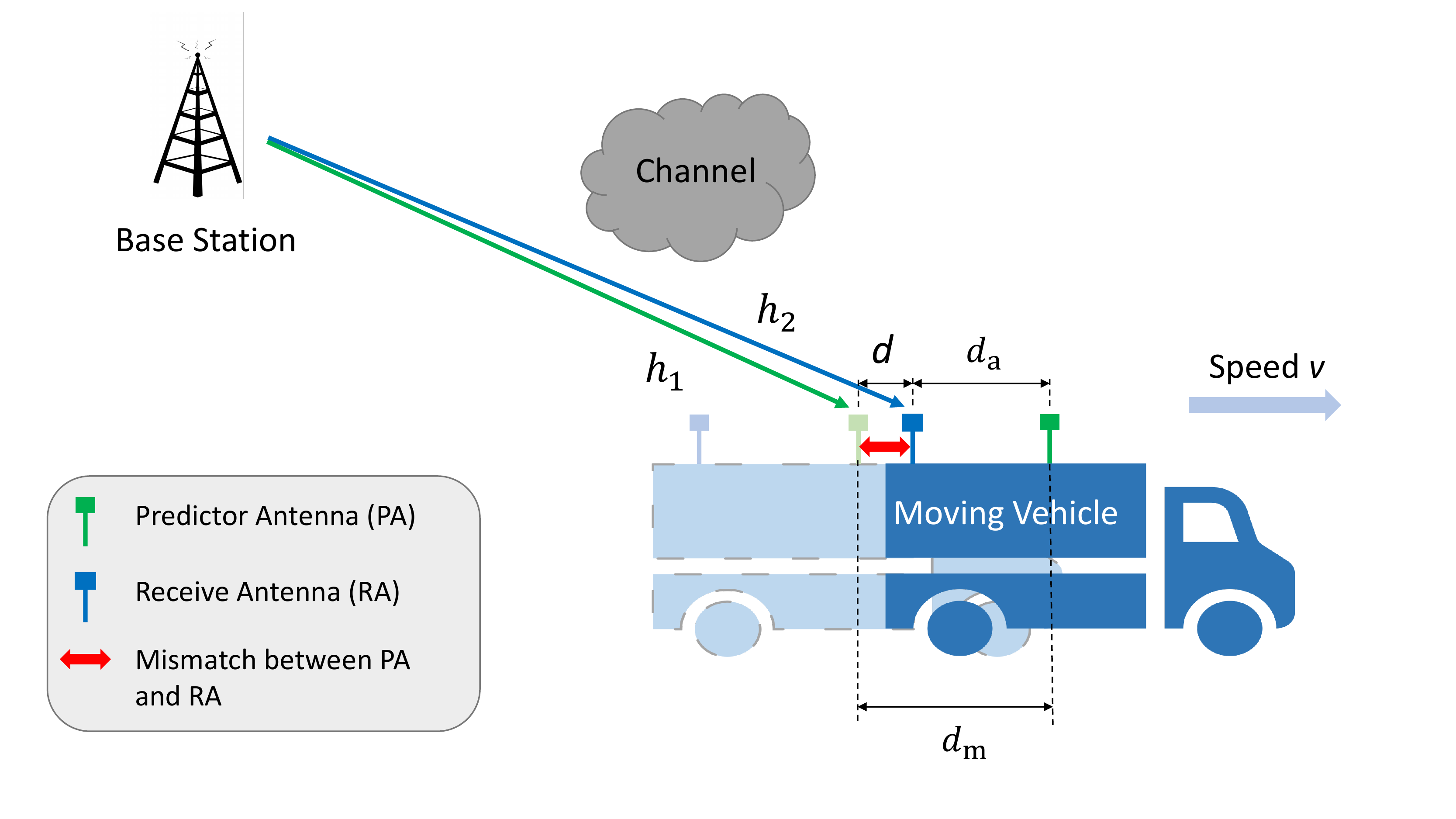}\\
\caption{Predictor antenna system with mismatch problem.}\vspace{-4mm}\label{ARQmodel}
\end{figure}
Figure \ref{ARQmodel} shows the standard PA system with two antennas on the roof of a vehicle. Here, the PA first sends pilots in time $t$. Then, the base station (BS) estimates the PA-BS channel $h_1$ and sends the data in time $t+\delta$ to the RA where  $\delta$ depends on the processing time at the BS. At the same time, the vehicle moves forward $d_\text{m}$ while the antenna separation between the PA and the RA is $d_\text{a}$. Then, considering downlink transmission in the BS-RA link, the  signal received by the RA is
\begin{align}\label{eq_Y}
y = \sqrt{P}h_2 x + z.
\end{align}
Here, $P$ represents the transmit power, $x$ is the input message with unit variance, and $h_2$ is the fading coefficient between the BS and the RA. Also, $z \sim \mathcal{CN}(0,1)$ denotes the independent and identically distributed (IID) complex Gaussian noise added at the receiver.

We represent the probability density function (PDF) and cumulative density function (CDF) of a random variable $A$ by $f_A(\cdot)$ and $F_A(\cdot)$, respectively.  Due to spatial mismatch between the PA and the RA, assuming a semi-static propagation environment, i.e., assuming that the coherence time of the propagation environment is much larger than $\delta$ \footnote{This has been experimentally verified in, e.g., \cite{Jamaly2014EuCAP}}, $h_2$ and $h_1$ are correlated according to \cite[Eq. 5]{Guo2019WCLrate} 
\begin{align}\label{eq_H}
    h_2 = \sqrt{1-\sigma^2} h_1 + \sigma q,
\end{align}
where $q \sim \mathcal{CN}(0,1)$ which is independent of the known channel value $h_1\sim \mathcal{CN}(0,1)$, and $\sigma$ is a function of the mis-matching distance $d = |d_\text{a}-d_\text{m}|$ \cite[Eq. 4]{Guo2019WCLrate}. Defining $g_1 = |h_1|^2$ and $ g_2 = |h_2|^2$, the CDF $F_{g_2|g_1}$ is given by
\begin{align}\label{eq_cdf}
    F_{g_2|g_1}(x) = 1 - Q_1\left( \sqrt{\frac{2g_1(1-\sigma^2)}{\sigma^2}}, \sqrt{\frac{2x}{\sigma^2}}  \right),
\end{align}
where $Q_1(s,\rho) = \int_{\rho}^{\infty} xe^{-\frac{x^2+s^2}{2}}I_0(xs)\text{d}x$, $s, \rho \ge 0$, is the first-order Marcum $Q$-function. Also, $I_n(x) = (\frac{x}{2})^n \sum_{i=0}^{\infty}\frac{(\frac{x}{2})^{2i} }{i!\Gamma(n+i+1)}$ is the $n$-order modified Bessel function of the first kind, and $\Gamma(z) = \int_0^{\infty} x^{z-1}e^{-x} \mathrm{d}x$ represents the Gamma function. In this way, although parameter adaptation is performed based on perfect CSIT of $h_1$ at time $t$, the spatial mismatch may lead to unsuccessful decoding by the RA  at $t+\delta$.

\vspace{-4mm}
\subsection{Proposed HARQ-based PA System}
Along with the spatial mismatch problem, the typical PA system  still suffers from poor spectral efficiency, compared to regular multiple-antenna system in static conditions, because the PA is used only for channel estimation. On the other hand, because the PA system includes the PA-BS feedback link,  HARQ can be supported by the PA structure. For this reason, we propose a setup as follows.

\begin{figure}
\centering
  \includegraphics[width=0.8\columnwidth]{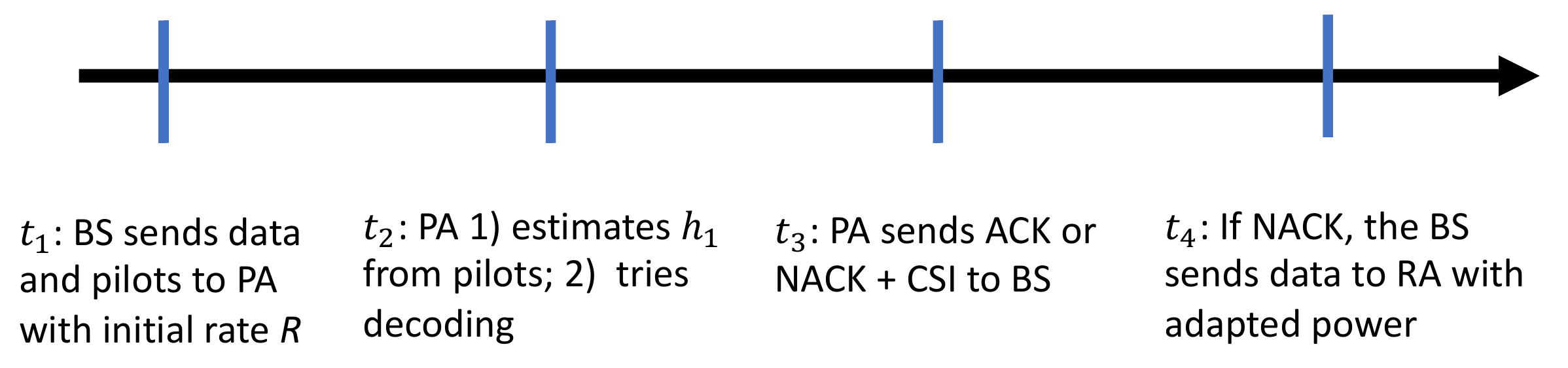}\\
\caption{Time structure for the proposed PA-HARQ scheme.}\vspace{-5mm}\label{fig_timeslot}
\end{figure}

Here, as  seen in Fig. \ref{fig_timeslot}, with no CSIT, at $t_1$ the BS sends pilots as well as the encoded data with certain initial rate $R$ and power $P_1$ to the PA. At $t_2$, the PA estimates the channel $h_1$ from the received pilots. At the same time, the PA tries to decode the signal. If the message is correctly decoded, i.e., $R\leq \log(1+g_1P_1)$,  an acknowledgment (ACK) is fed back to the BS at $t_3$, and the data transmission stops. Otherwise, the PA sends both a negative acknowledgment (NACK) and high accuracy quantized CSI feedback about $h_1$. The number of quantization bits are large enough such that we can assume the BS to have perfect CSIT of $h_1$ (see \cite{guo2020semilinear} for the effect of imperfect CSIT on the performance of PA systems). With NACK, in the second transmission round at time $t_4$, the BS  transmits the message to the RA with power $P_2$  which is a function of the instantaneous channel quality $g_1$. The outage occurs if the RA cannot decode the message at the end of the second round. 
\vspace{-3mm}
\section{Analytical Results}
Let $\epsilon$ be the outage probability constraint. Here, we present the results for the cases with RTD and INR HARQ protocols. With an RTD protocol, the same signal (with possibly different power) is sent in each retransmission round, and the receiver performs maximum ratio combining of all received copies of the signal. With INR, on the other hand, new redundancy bits are sent in the retransmissions, and the receiver decodes the message by combining all received signals \cite{6566132,5754756}.

Considering Rayleigh fading conditions with $f_{g_1}(x) =  e^{- x}$, the outage probability at the end of Round 1 is given by
\begin{align}\label{eq_pout}
&\text{Pr}(\text{Outage, Round 1}) = \text{Pr}\left\{R\leq \log(1+g_1P_1)\right\}\nonumber\\
&~~ = \text{Pr}\left\{g_1\leq \frac{e^{R}-1}{P_1}\right\}
 = 1-e^{ -\frac{\theta}{P_1}}, 
\end{align}
where $\theta = e^{R}-1$. Then, using the results of, e.g.,  \cite[Eq. 7, 18]{6566132} on the outage probability of the RTD- and INR-based HARQ protocols, the power allocation problem for the proposed HARQ-based PA system can be stated as
\begin{equation}
\label{eq_optproblem}
\begin{aligned}
\min_{P_1,P_2} \quad & \mathbb{E}_{g_1}\left[P_\text{tot}|g_1\right] \\
\textrm{s.t.} \quad &  P_1, P_2 > 0,\\
&P_\text{tot}|g_1 = \left[P_1 + P_2(g_1) \times \mathcal{I}\left\{g_1 \le \frac{\theta}{P_1}\right\}\right],
\end{aligned}
\end{equation}
with
\begin{align}\label{eq_optproblemrtd}
F_{g_2|g_1}\left\{\frac{\theta-g_1P_1}{P_2(g_1)}    \right\} = \epsilon, \quad\text{for RTD}
\end{align}
\begin{align}\label{eq_optprobleminr}
F_{g_2|g_1}\left\{\frac{e^{R-\log(1+g_1P_1)}-1}{P_2(g_1)}    \right\} = \epsilon, \quad\text{for INR}. 
\end{align}
Here, $P_\text{tot}|g_1$ is the total instantaneous transmission power for two transmission rounds (i.e., one retransmission) with given $g_1$, and we define $\Bar{P} \doteq \mathbb{E}_{g_1}\left[P_\text{tot}|g_1\right]$ as the expected power, averaged over $g_1$. Moreover, $\mathcal{I}(x)=1$ if $x>0$ and $\mathcal{I}(x)=0$ if $x \le 0$. Also, $\mathbb{E}_{g_1}[\cdot]$ represents the expectation operator over $g_1$. Here, we ignore the peak power constraint and assume that the BS is capable of transmitting  sufficiently high power. Finally, (\ref{eq_optproblem})-(\ref{eq_optprobleminr}) come from the fact that, with our proposed scheme, $P_1$ is fixed and optimized with no CSIT at the BS and based on average system performance. On the other hand, $P_2$ is adapted continuously based on the predicted CSIT.

Using (\ref{eq_optproblem}), the required power in Round 2 is given by
\begin{equation}
\label{eq_PRTDe}
    P_2(g_1) = \frac{\theta-g_1P_1}{F_{g_2|g_1}^{-1}(\epsilon)},
\end{equation}
for the RTD, and
\begin{equation}
\label{eq_PINRe}
    P_2(g_1) = \frac{e^{R-\log(1+g_1P_1)}-1}{F_{g_2|g_1}^{-1}(\epsilon)},
\end{equation}
for the INR, where $F_{g_2|g_1}^{-1}(\cdot)$ is the inverse of the CDF given in (\ref{eq_cdf}). Note that, $F_{g_2|g_1}^{-1}(\cdot)$ is a complex function of $g_1$ and, consequently, it is not possible to express $P_2$ in closed-form. For this reason, one can use   \cite[Eq. 2, 7]{6414576}
\begin{align}
    Q_1 (s, \rho) &\simeq e^{\left(-e^{\mathcal{I}(s)}\rho^{\mathcal{J}(s)}\right)}, \nonumber\\
    \mathcal{I}(s)& = -0.840+0.327s-0.740s^2+0.083s^3-0.004s^4,\nonumber\\
    \mathcal{J}(s)& = 2.174-0.592s+0.593s^2-0.092s^3+0.005s^4,
\end{align}
to approximate $F_{g_2|g_1}$ and consequently $F_{g_2|g_1}^{-1}(\epsilon)$. In this way, (\ref{eq_PRTDe}) and (\ref{eq_PINRe}) can be approximated as
\begin{equation}
\label{eq_PRTDa}
    P_2(g_1) = \Omega\left(\theta-g_1P_1\right),
\end{equation}
for the RTD, and
\begin{equation}
\label{eq_PINRa}
   P_2(g_1) = \Omega\left(e^{R-\log(1+g_1P_1)}-1\right),
\end{equation}
for the INR, where
\begin{equation}\label{eq_omega}
    \Omega (g_1) = \frac{2}{\sigma^2}\left(\frac{\log(1-\epsilon)}{-e^{\mathcal{I}\left(\sqrt{\frac{2g_1(1-\sigma^2)}{\sigma^2}}\right)}}\right)^{-\frac{2}{\mathcal{J}\left(\sqrt{\frac{2g_1(1-\sigma^2)}{\sigma^2}}\right)}}.
\end{equation}

In this way, for different HARQ protocols, we can express the  instantaneous transmission power of Round 2, for every given $g_1$ in closed-form. Then, the  power allocation problem (\ref{eq_optproblem}) can be solved numerically. However, (\ref{eq_omega}) is still complicated and it is not possible to solve (\ref{eq_optproblem}) in closed-form. For this reason, we propose an approximation scheme to solve (\ref{eq_optproblem}) as follows.

Let us initially concentrate on the RTD protocol. Then,  combining (\ref{eq_optproblem}) and (\ref{eq_PRTDe}), the expected total transmission power  is given by
\begin{align}\label{eq_barP}
    \Bar{P}_\text{RTD} &= P_1 + \int_0^{\theta/P_1} e^{- x}P_2\text{d}x\nonumber\\
    &= P_1 + \int_0^{\theta/P_1} e^{- x}\frac{\theta-x P_1}{F_{g_2|x}^{-1}(\epsilon)}\text{d}x.
\end{align}
Then, Theorem \ref{theorem1} derives the minimum required power  in Round 1 and the average total power consumption as follows.

\begin{theorem}\label{theorem1}
With RTD and given outage constraint $\epsilon$, the minimum required power in Round 1 and the average total power are given by (\ref{eq_P1}) and (\ref{eq_666666}), respectively. 
\end{theorem}
\begin{proof}
Plugging (\ref{eq_cdf}) into (\ref{eq_optproblemrtd}), we have
\begin{align}
    1-Q_1\left(\sqrt{\frac{2g_1(1-\sigma^2)}{\sigma^2}},\sqrt{\frac{2(\theta-g_1P_1)}{\sigma^2 P_2}}\right) = \epsilon.
\end{align}
By using the approximation \cite[Eq. 17]{Azari2018TCultra} for moderate/large $\sigma$, i.e., if $1-Q_1(s,\rho) = 1-\epsilon$, then $\rho = Q_1^{-1}(s, 1-\epsilon) \simeq \sqrt{-2\log(1-\epsilon)}e^{\frac{s^2}{4}}$, we can obtain 
\begin{align}
    \sqrt{\frac{2(\theta-g_1P_1)}{\sigma^2 P_2}}\simeq \sqrt{-2\log(1-\epsilon)}e^{\frac{g_1(1-\sigma^2)}{2\sigma^2}}.
\end{align}
In this way, $P_2$ in (\ref{eq_barP}) is approximated by
\begin{align}
    P_2 \simeq (\theta - g_1P_1)\frac{e^{-\frac{g_1(1-\sigma^2)}{\sigma^2}}}{-\sigma^2\log(1-\epsilon)},
\end{align}
and considering RTD, (\ref{eq_barP}) can be rewritten as
\begin{align}\label{eq_Papprobeforederrivative}
    \Bar{P} &= P_1 + \int_0^{\theta/P_1} e^{- x}(\theta - xP_1)\frac{e^{-\frac{x(1-\sigma^2)}{\sigma^2}}}{-\sigma^2\log(1-\epsilon)}\text{d}x \nonumber\\
    &\overset{(a)}{=} P_1 + \frac{c}{m^2}\left(P_1 e^{-\frac{m\theta}{P_1}}-P_1 + m\theta\right),
\end{align}
where, in (a) we set $m = 1 + \frac{1-\sigma^2}{\sigma^2}$ and $c = \frac{-1}{\sigma^2\log(1-\epsilon)}$ for simplicity. Then, setting the derivative of (\ref{eq_Papprobeforederrivative})  with respect to $P_1$ equal to zero, the minimum $P_1$ for the minimum total power can be found as
\begin{align}
  \label{eq_P1}
    \hat{P}_{1,\text{RTD}} & = \operatorname*{arg}_{P_1 > 0} \Bigg\{ 1 +  \frac{c}{m^2}e^{-\frac{\theta m}{P_1}}\left(\frac{m\theta}{P_1}+1\right) - \frac{c}{m^2} = 0\Bigg\}\nonumber\\
    & = \operatorname*{arg}_{P_1 > 0} \left\{e^{-\frac{\theta m}{P_1}}\left(\frac{m\theta}{P_1}+1\right) = 1-\frac{m^2}{c} \right\}\nonumber\\
    & \overset{(b)}= \frac{-m\theta}{\mathcal{W}_{-1}\left(\frac{m^2}{ce}-\frac{1}{e}\right)+1}.
\end{align}
Here, $(b)$ is obtained by the definition of the Lambert W function $xe^x = y \Leftrightarrow x = \mathcal{W}(y)$ \cite{corless1996lambertw}. Also, because $\frac{m^2}{ce}-\frac{1}{e}<0$, we use the $\mathcal{W}_{-1}(\cdot)$ branch \cite[Eq. 16]{veberic2010having}. Then, plugging (\ref{eq_P1}) into (\ref{eq_Papprobeforederrivative}), we obtain the  minimum total transmission power as
\begin{align}\label{eq_666666}
    \hat{\Bar{P}}_{\text{RTD}} =\hat{P}_{1,\text{RTD}} + \frac{c}{m^2}\left(\hat{P}_{1,\text{RTD}} e^{-\frac{m\theta}{\hat{P}_{1,\text{RTD}}}}-\hat{P}_{1,\text{RTD}} + m\theta\right).
\end{align}
\end{proof}

\vspace{-12mm}
\subsection{On the Effect of CSI Feedback/Power Allocation}
As a benchmark, in this part, we consider the case without exploiting CSIT, i.e., we consider the typical HARQ schemes where CSI feedback is not sent along with NACK, and we do not perform power adaptation.  Here, the outage probability,  for the RTD and the INR  are given by
\begin{align}\label{eq_PrRTD}
    \zeta_{\text{RTD}} = \text{Pr}\left\{\log\left(1+\left(g_1+g_2\right)P\right)<R\Big|g_1<\frac{\theta}{P}  \right\},
\end{align}
\begin{align}\label{eq_PrINR}
\begin{split}
    \zeta_{\text{INR}}=\text{Pr}\left\{\log\left(1+g_1P\right) + \log\left(1+g_2P\right)<R\Big|g_1<\frac{\theta}{P}\right\}.
\end{split}
\end{align}
Also, in both protocols, the total average power is given by
\begin{align}\label{eq_barpbench}
    \Bar{P} &= P + P\cdot\text{Pr}\{\log\left(1+g_1P\right)<R\}\nonumber\\
    &=P(2-e^{-\frac{\theta}{P}}).
\end{align}
\begin{theorem}\label{theorem2}
Without CSIT feedback/power allocation, the outage probability of the RTD-based PA-HARQ scheme is given by (\ref{eq_RTDnopa}).
\end{theorem}
\begin{proof}
Using (\ref{eq_PrRTD}), we have
\begin{align}
    \zeta_{\text{RTD}} &= \text{Pr}\left\{g_2< \frac{\theta}{P}-g_1  \Big| g_1 < \frac{\theta}{P} \right\}\nonumber\\
    &= \frac{1}{1-e^{-\frac{\theta}{P}}}\int_{0}^{\frac{\theta}{P}} e^{- x} F_{g_2|x}\left(\frac{\theta}{P}-x\right) \text{d}x. \label{eq_nocloseform}
\end{align}
Considering (\ref{eq_cdf}), there is no closed-form solution for (\ref{eq_nocloseform}).  For this reason, we use the  approximation \cite[Eq. 14]{andras2011generalized} 
\begin{align}
    Q_1\left( s, \rho  \right) \simeq  1 - \left( 1 + \frac{s^2}{2} \right)e^{-\frac{s^2}{2}} +
    \left( 1 + \frac{s^2}{2} + \frac{s^2\rho^2}{4} \right)e^{-\frac{s^2+\rho^2}{2}},
\end{align}
which simplifies (\ref{eq_nocloseform}) to 
\begin{align}\label{eq_RTDnopa}
    &\zeta_{\text{RTD}} \simeq 
    -\frac{e^{-\frac{\theta}{P\sigma^2}}}{6\sigma^4\left(1-e^{-\frac{\theta}{P}}\right)}\Bigg(\left(6\sigma^8-12\sigma^6\right)e^{\frac{\theta}{P\sigma^2}}+12\sigma^6+\nonumber\\&\left(3\sigma^2-3\sigma^4\right)\frac{\theta^2}{P^2}+\left(12\sigma^4-6\sigma^6\right)\frac{\theta}{P}+\left(1-\sigma^2\right)\frac{\theta^3}{P^3}-6\sigma^8\Bigg)
\end{align}

\end{proof}
\vspace{-12mm}
\subsection{On the Effect of Introducing INR}
For the INR scheme, by using Jensen's inequality and the concavity of the logarithm function \cite[Eq. 30]{makki2016TWCRFFSO}
\begin{align}\label{eq_jensen}
   \frac{1}{n}\sum_{i=1}^{n} \log (1+x_i)\leq\log\left(1+\frac{1}{n}\sum_{i=1}^{n}x_i\right),
\end{align}
the closed-form expressions for the minimum required power, the average total power, as well as outage probability without power allocation are given by the following Corollary.
\begin{corollary}
With INR, the minimum required power in Round 1 and the average total power are given by (\ref{eq_P1inr}) and (\ref{eq_barpinr}), respectively.  Also, without power allocation, the outage probability of the INR-based scheme is given by (\ref{eq_INRnopa}).
\end{corollary}
\begin{proof}
Using (\ref{eq_optprobleminr}), the Jensen's inequality (\ref{eq_jensen}) and defining $\theta_1 =2\left(e^{\frac{R}{2}}-1\right)$, (\ref{eq_PINRe}) can be approximated by
\begin{align}
     P_2(g_1) \simeq \frac{\theta_1 - g_1P_1}{F_{g_2|g_1}^{-1}(\epsilon)},
\end{align}
and, following the same steps as in Theorem \ref{theorem1}, we obtain the minimum required $P_1$ for the minimum total power in the INR scheme as
\begin{align}
  \label{eq_P1inr}
    \hat{P}_{1,\text{INR}}= \frac{-m\theta_1}{\mathcal{W}_{-1}\left(\frac{m^2}{ce}-\frac{1}{e}\right)+1}.
\end{align}
Also, the  minimum total  power can be calculated by
\begin{align}\label{eq_barpinr}
    \hat{\Bar{P}}_{\text{INR}} =\hat{P}_{1,\text{INR}} + \frac{c}{m^2}\left(\hat{P}_{1,\text{INR}} e^{-\frac{m\theta_1}{\hat{P}_{1,\text{INR}}}}-\hat{P}_{1,\text{INR}} + m\theta_1\right).
\end{align}

Finally, (\ref{eq_PrINR}) can be further derived by
\begin{align}
    \zeta_{\text{INR}} &= \text{Pr}\left\{  g_2< \frac{e^{R-\log(1+g_1)}-1}{P}\Big|g_1 < \frac{\theta}{P} \right\}\nonumber\\
    &\overset{(c)}{\simeq}\frac{1}{1-e^{-\frac{\theta}{P}}}\int_{0}^{\frac{\theta}{P}} e^{- x} F_{g_2|x}\left(\frac{\theta_1}{P}-x\right) \text{d}x,
\end{align}
with $(c)$ using Jensen's inequality. Then, following the same steps as in Theorem \ref{theorem2}, the outage probability is found as 
\begin{align}\label{eq_INRnopa}
    \zeta_{\text{INR}} \simeq \zeta_{\text{RTD}}(\theta = \theta_1).
\end{align}
\end{proof}\vspace{-4mm}

Finally, as a benchmark, we consider the case with no retransmission where the outage probability is given by $1-e^{-\frac{\theta}{P}}$. In this case, the required outage-constrained power without retransmission is given by $P \geq \frac{1-e^{R}}{\log(1-\epsilon)}$.




\section{{Simulation Results}}

In the simulations, we set $ \delta$ = 5 ms, $f_\text{c}$ = 2.68 GHz, and $d_\text{a} = 1.5\lambda$. Each point in the figures is obtained by averaging the system performance over $1\times10^5$ channel realizations. 

\begin{figure}
\centering
  \includegraphics[width=0.9\columnwidth]{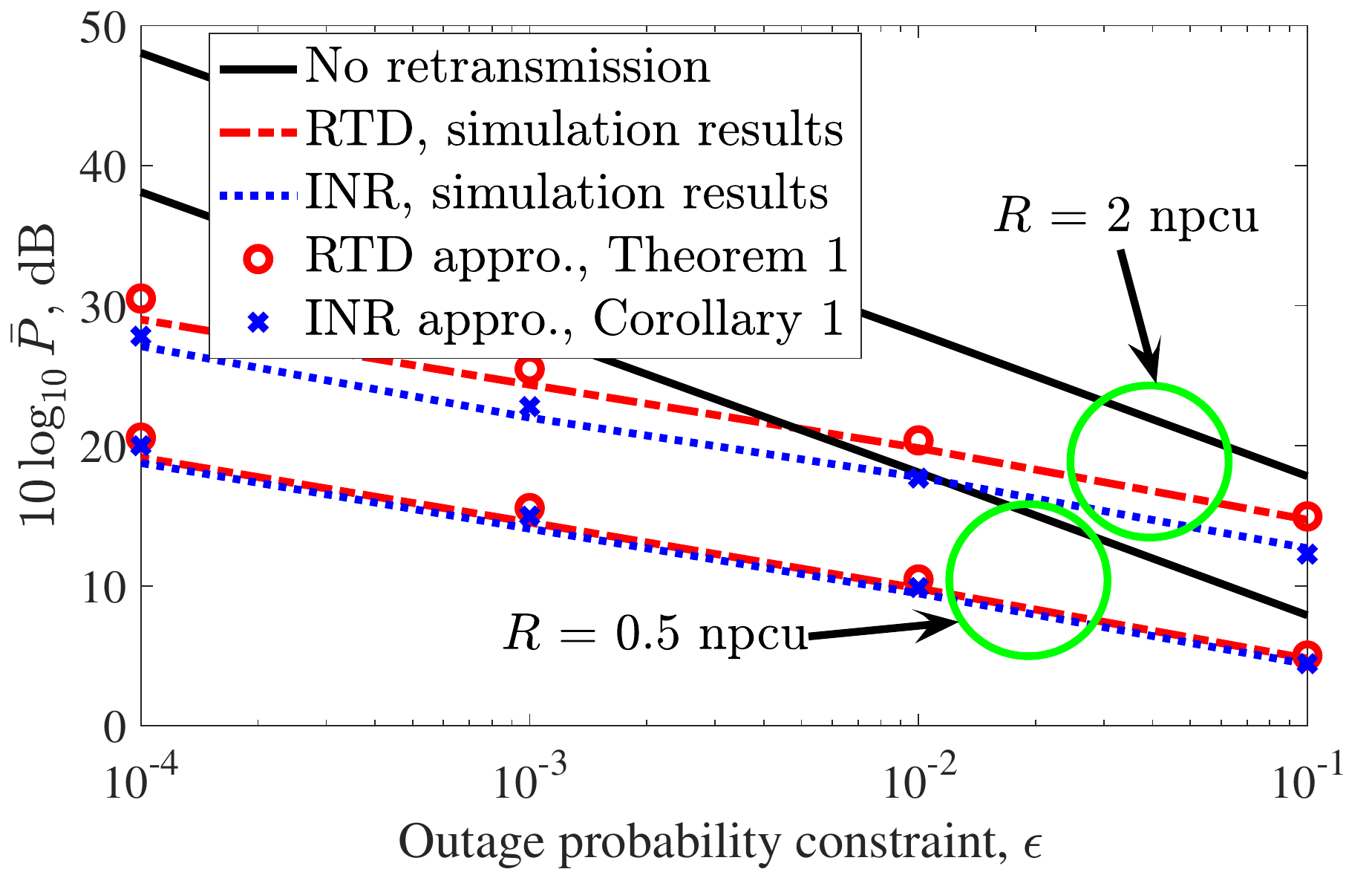}\\
\caption{Minimum required power under various outage constraints for the RTD and the INR, $\sigma = 0.8$, $R = 0.5, 2$ npcu.}\vspace{-3mm}\label{figure3}
\end{figure}

\begin{figure}
\centering
  \includegraphics[width=0.9\columnwidth]{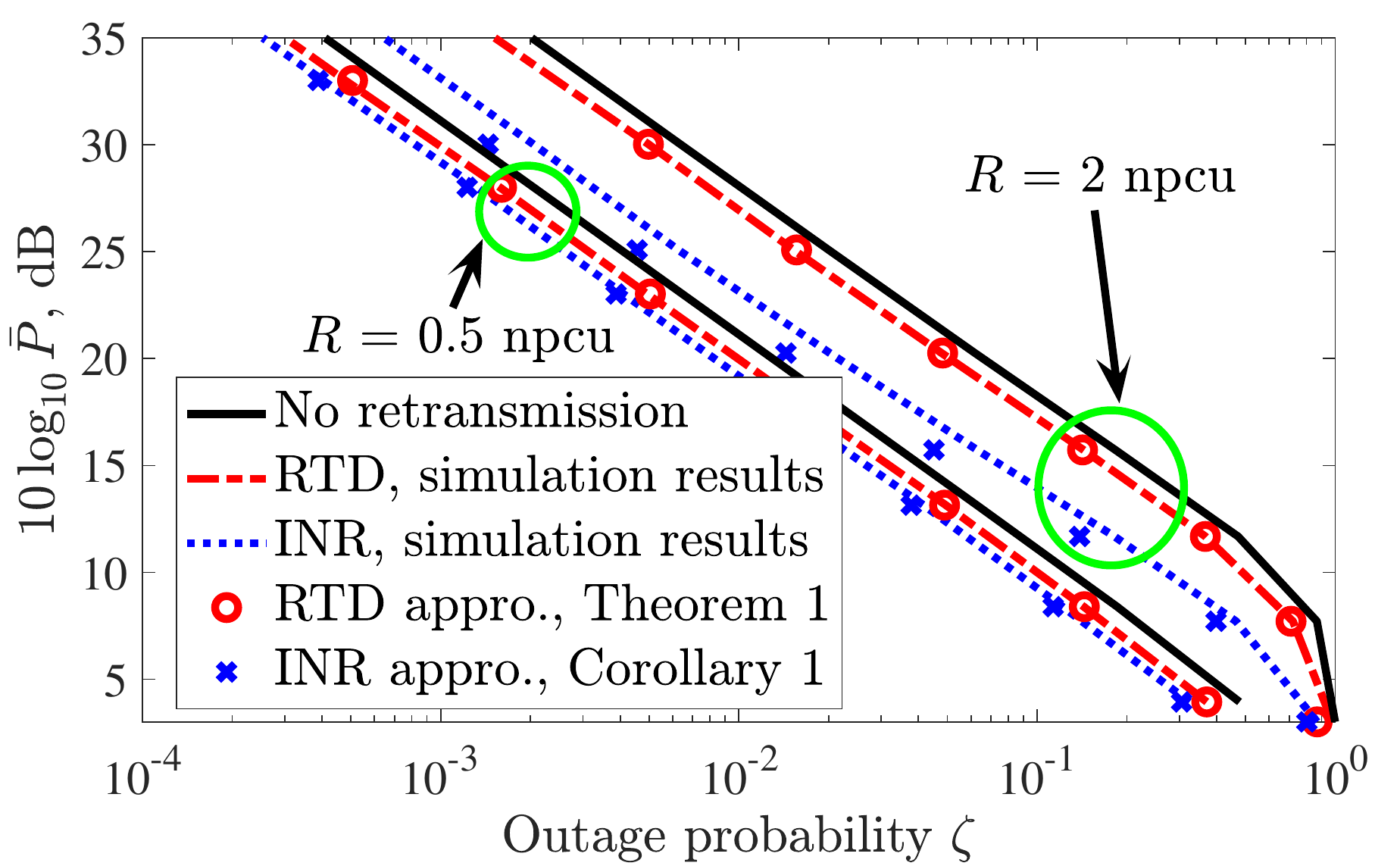}\\
\caption{Minimum required power without optimized power allocation (\ref{eq_barpbench}) under various outage probabilities, $\sigma$ = 0.8, $R$ = 0.5, 2 npcu.}\vspace{-5mm}\label{figure4}
\end{figure}

\begin{figure}
\centering
  \includegraphics[width=0.9\columnwidth]{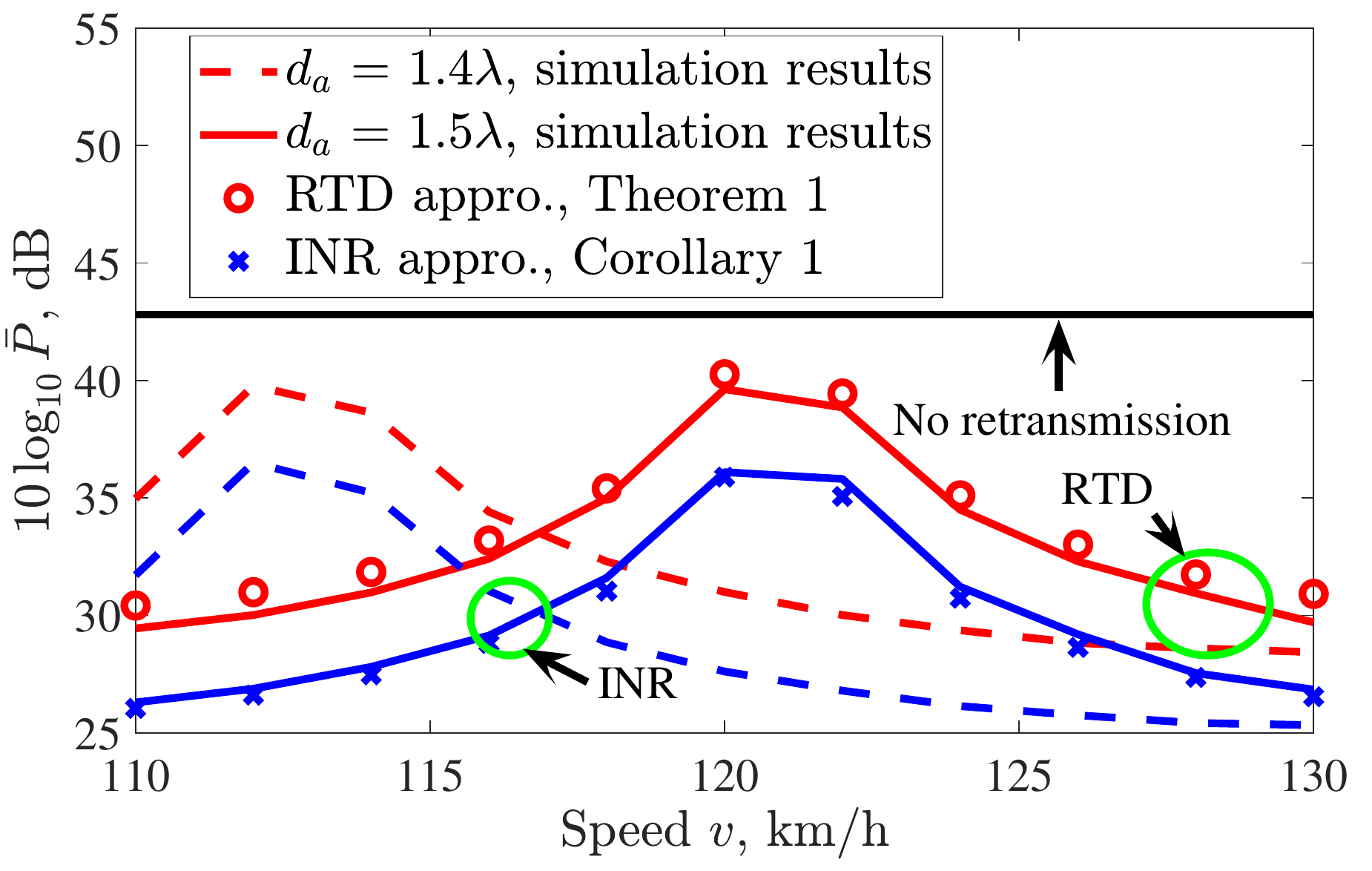}\\
\caption{Minimum required power for the RTD and the INR with given outage constraint, for different speed $v$. $R = 3$ npcu, $\epsilon = 10^{-3}$.}\vspace{-5mm}\label{figure5}
\end{figure}

Figure \ref{figure3} shows the results of the power allocation problem (\ref{eq_optproblem}) for both the RTD and the INR with different initial rates $R = 0.5, 2$ npcu, $\sigma = 0.8$, and different outage probability constraints $\epsilon$. Here, the simulation results are obtained from  solving the optimization problem (\ref{eq_optproblem}) numerically, while the approximation results for the RTD and the INR are obtained by Theorem \ref{theorem1} and Corollary 1, respectively.  In Fig. \ref{figure4}, we plot the outage probability using the  RTD and the INR  without power optimization. Here, we set $\sigma = 0.8$, and the initial rate $R = 0.5, 2$ npcu. The simulation results are obtained by (\ref{eq_PrRTD}) and (\ref{eq_PrINR}), while the analytical approximations are obtained from Theorem \ref{theorem2} for the RTD and Corollary 1 for the INR. Finally, in Fig. \ref{figure5},  we study the minimum required transmission power for different speeds $v$. Here, manipulating $v$ is  equivalent to changing the level of spatial correlation for given values of $\delta$, $f_\text{c}$ and $d_\text{a}$ (see \cite{guo2020semilinear} for the detailed effect of the vehicle speed on the spatial correlation). Also, we study the effect of different values of $d_\text{a}$ on the system performance.   According to the figures, we can conclude the following:




\begin{itemize}
    \item The approximation schemes of Theorem \ref{theorem1} and Corollary 1 are  tight for a broad range of values of initial rate $R$, speed $v$,  as well as $\epsilon$ (Figs. \ref{figure3} and \ref{figure5}). Thus, for different parameter settings, the outage-constrained power allocation for the RTD and the INR can be well determined by Theorem \ref{theorem1} and Corollary 1. 
    \item Also, the approximation scheme of Theorem \ref{theorem2} is tight for a broad range of values of average power (\ref{eq_barpbench}) as well as initial rate $R$ (Fig. \ref{figure4}). Thus, for different parameter settings, the outage probability for the proposed PA-HARQ scheme without power allocation, can be well determined by Theorem \ref{theorem2} and Corollary 1.
    \item With the deployment of the PA and power allocation, remarkable power gain is achieved especially at low outage probabilities  (Figs. \ref{figure3} and \ref{figure5}). Moreover,  as also indicated in, e.g.,  \cite{6566132}, INR outperforms RTD, in terms of outage-limited average power. However, the difference between the performance of these protocols decreases as the data rate or $\sigma$ decreases (Figs. \ref{figure3} and \ref{figure5}).
    \item Figures \ref{figure3}-\ref{figure5} emphasize the efficiency of HARQ as well as adaptive power allocation in the PA system. From Fig. \ref{figure4}, we can see that with the PA-HARQ setup, even without power allocation, we can obtain considerable performance improvements compared to the case with no retransmission. Then, as seen in Fig. \ref{figure3}, with power allocation there is much larger gain for the system performance. For instance, with $\epsilon=10^{-5}$ and $R = 4$ npcu, the RTD- and INR-based PA-HARQ can reduce the required power, compared to no retransmission case, by 25 and 30 dB, respectively. Then, as can be seen from Fig.  \ref{figure5}, the effect of power-optimized PA-HARQ increases with $\sigma$. This is because the larger $\sigma$ provides better spatial diversity of the channel, which can improve the performance with retransmissions. For the peaks in Fig. \ref{figure5}, the channel for the retransmission has the largest correlation with the one in the first round, which leads to the smallest power gain. Moreover, when the antenna separation $d_\text{a}$ decreases, the speed where the power gain is minimum, also decreases, due to the reduction of the mismatch distance $d$.
\end{itemize}

\vspace{-3mm}
\section{Conclusion}
We studied the PA-HARQ system with the spatial mismatch problem, in the context of outage-constrained power allocation. We derived  closed-form expressions for the minimum instantaneous and  total transmit power. The approximations are tight for a broad range of system configurations. Also, the  results show that, while PA-assisted adaptive power adaptation leads to considerable performance improvements, the total transmission power and the outage probability are remarkably affected by the spatial mismatch.
\vspace{-2mm}





\bibliographystyle{IEEEtran}

\bibliography{main.bib}

\end{document}